\documentclass[aps,onecolumn,twoside,floatfix,pra,a4paper,11pt]{revtex4-2} 

\usepackage{times}
\usepackage{epsfig}
\usepackage{amsfonts}
\usepackage{amsmath}
\usepackage{amssymb,amsthm}
\usepackage{xcolor,colortbl}
\usepackage{multirow}
\usepackage{braket}
\usepackage[normalem]{ulem}
\usepackage{latexsym}
\usepackage{tabularx}
\usepackage{amsfonts}
\usepackage{mathrsfs}
\usepackage{natbib}
\usepackage{verbatim}
\usepackage{gensymb}
\usepackage{caption}
\usepackage{caption}
\usepackage{subcaption}
\usepackage{ragged2e}
\DeclareCaptionJustification{justified}{\justifying}
\usepackage{blkarray}
 \usepackage{graphicx}
 \usepackage[export]{adjustbox}

 \newcommand{\ba}{\begin{eqnarray}}
\newcommand{\ea}{\end{eqnarray}}

\newcommand{\ie}{{\it i.e. }}

\DeclareMathOperator\Tr{Tr}

\newtheorem{definition}{Definition}
\newtheorem{proposition}{Proposition}

\usepackage[colorlinks=true,linkcolor=blue,citecolor=magenta,urlcolor=blue]{hyperref}
\allowdisplaybreaks

\begin{document}
\raggedbottom

\title{Certification of Absolutely Entangled Sets in the Prepare-and-Measure Scenario}
\author{Ram Krishna Patra}
\affiliation{HUN-REN Institute for Nuclear Research, PO Box 51, H-4001 Debrecen, Hungary}

\author{Abdelmalek Taoutioui}
\affiliation{HUN-REN Institute for Nuclear Research, PO Box 51, H-4001 Debrecen, Hungary}

\author{Tamás Vértesi}
\affiliation{HUN-REN Institute for Nuclear Research, PO Box 51, H-4001 Debrecen, Hungary}
\begin{abstract}
    An absolutely entangled set (AES) is a collection of states for which at least one member remains entangled under any choice of global unitary transformation. We develop a semi-device-independent framework for certifying this property from dimensionally bounded prepare-and-measure correlations. A witness value above the maximum attainable with separable preparations of a given dimension rules out every non-AES realization. To upperbound this limit, we construct a semidefinite-programming (SDP) hierarchy tailored to separable preparations. For a two-qubit random-access-code witness, the SDP upper bound and the heuristic lower bound coincide to the reported precision, yielding a concrete AES certification. We also formulate a general construction of witnesses for extended preparation sets containing target pure-state AESs. Examples based on equiangular tight frames, a minimal AES, and mutually unbiased bases further illustrate the scope of the method. In summary, our results establish a systematic and experimentally feasible approach to the certification of AESs in semi-device-independent quantum information protocols.
\end{abstract}
\maketitle

\section{Introduction}
Quantum entanglement is one of the most fundamental features of quantum mechanics and constitutes the central resource underlying the advantages of quantum information processing over its classical counterparts. In the seminal quantum teleportation protocol, the assistance of a maximally entangled state and two bits of classical communication enables the perfect transfer of an unknown quantum state between distant parties \cite{Bennett1993}. Subsequently, it was demonstrated that every entangled state can exhibit nonclassical teleportation capability \cite{Cavalcanti2017}. Entanglement also provides a significant advantage in communication tasks, as exemplified by quantum superdense coding, where the communication capacity of a single qubit can be effectively doubled through shared entanglement \cite{Bennett1992}. Beyond these paradigmatic applications, quantum entanglement plays a crucial role in a broad range of quantum information protocols, including communication complexity reduction \cite{Cleve1997,Buhrman2010}, quantum key distribution \cite{Bennett1984BB84,Ekert1991,Cerf2002,Gisin2002}, randomness certification \cite{Acin2012,Wooltorton2022,Skrzypczyk2018}, and several other quantum-enhanced information-processing tasks.

In general, entanglement is inherently subsystem dependent, since the same composite quantum system may appear either entangled or separable depending on the choice and definition of the subsystems \cite{Zanardi2001,Zanardi2004}. More precisely, the definition of subsystems, and consequently that of entanglement, depends on the choice of global reference frame, which can be equivalently described by the action of a global unitary transformation \(U\) on the composite quantum system. As a result, any entangled state can, in principle, be transformed into a product state under an appropriate global unitary operation. Consequently, individual entangled states are generally insufficient for establishing reference-frame-independent entanglement advantage. To address this limitation, recent works have introduced the concept of sets of states for which at least one member remains entangled under arbitrary global unitary transformations \cite{Cai2021,Yu2021,Li2020arxiv}. Such collections, known as AES, provide a fundamental framework for realizing reference-frame-independent entanglement-based quantum advantages.

On the other hand, another prominent communication framework, namely the prepare-and-measure (PM) scenario, plays a central role in quantum information processing tasks. In this setting, a preparation device encodes information into classical or quantum states, which are subsequently transmitted to a measurement device. The PM scenario has emerged as a well-established platform for the semi-device-independent (SDI) certification of quantum resources, relying solely on assumptions regarding the dimension of the underlying communicated system. Beyond dimension-bounded frameworks, several alternative SDI approaches based on constraints such as entropy \cite{Chaves2015}, distinguishability \cite{Brunner2013,Brask2017,Wang2019NPJQI}, information content \cite{Tavakoli2020quantum,Chaturvedi2020}, rotational symmetry \cite{Jones2026} and energy \cite{RochiCarceller2026,Raffaele2025} have also been investigated. For a recent review of PM scenarios and their SDI applications, see Ref.~\cite{Brask2026}. The PM framework has found numerous applications in quantum communication \cite{Gallego2010,Bowles2015}, quantum random access codes (QRACs) \cite{Ambainis2002,Tavakoli2015,M2021}, quantum random number generation (QRNG) \cite{Li2011,Lunghi2015,Passaro2015}, self-testing protocols \cite{Navascues2023,Drotos2024,Tavakoli2018,Miklin2021}, detection of quantum composition \cite{Naik2022,Patra2023} and various foundational aspects of quantum theory. Furthermore, several nonclassical quantum resources, including magic \cite{Chowdhury2026, Zamora2025}, measurement incompatibility \cite{Egelhaaf2025,Veeren2024}, and measurements \cite{Mironowicz2019,Farkas2019,Tavakoli2020ScAd}, have been successfully certified within PM scenarios. Despite the central role of entanglement in enabling quantum advantages, its certification in prepare-and-measure settings has thus far been largely restricted to entanglement-assisted prepare-and-measure (EAPM) scenarios \cite{Tavakoli2021,Bakhshinezhad2024}. In contrast, the certification of entanglement in the standard PM framework without presharing auxiliary entanglement remains largely unexplored.

Therefore, in this work, we address this gap by developing a semi-device-independent framework for certifying entanglement in a reference-frame-independent manner. To this end, we construct witnesses based on PM correlations that are bounded for all separable preparations. As a key technical contribution, we develop an SDP hierarchy specifically tailored to systematically upper-bound the correlations achievable with separable preparations, thereby providing rigorous separable bounds for the proposed witnesses. Consequently, under the assumption of a bounded Hilbert-space dimension and the causal structure of the PM setup, an observed violation of the corresponding separable bound certifies that the prepared set of states constitutes an AES. Beyond the certification of specific instances, we provide a general framework for constructing correlation witnesses for an extended set of AES constructed by augmenting the target AESs with ancillary states. Finally, we demonstrate the scope of our approach through several representative families of AESs arising from distinct state configurations.

\section{Absolutely Entangled Sets}
Let $[N]$ denote the set $\{1,2,\cdots,N\}$ for any positive integer $N$. We also employ the notation $\mathbb{C}^{d} \equiv \mathbb{C}^{d_1} \otimes \mathbb{C}^{d_2}$  and use them interchangeably, where $d=d_1 d_2$. Furthermore, we denote the set of all unitary operators acting on $\mathbb{C}^{d_1} \otimes \mathbb{C}^{d_2}$ by $U(d_1d_2)$. Entanglement is defined between two subsystems of a composite system, in general, it depends on the choice of the global reference frame, since a global unitary transformation can map an entangled state to a separable one. Consequently, reference-frame-independent quantum advantage can only be established by considering states that remain entangled under arbitrary global unitary transformations. This observation motivates the notion of AES \cite{Cai2021}, which is formally defined as 
\begingroup
\renewcommand{\thedefinition}{1a}
\begin{definition}
    A set of $m$ pure states $S=\{\ket{\psi_1},\ket{\psi_2},\cdots,\ket{\psi_m}\}$ in the Hilbert space $\mathbb{C}^{d_1}\otimes \mathbb{C}^{d_2}$ is said to form an AES if for every unitary transformation $U\in U(d_1d_2)$ at least one of the transformed states $U\ket{\psi_i}$ in the set remains entangled with respect to the bipartition $\mathbb{C}^{d_1}\otimes \mathbb{C}^{d_2}$.
\end{definition}
\endgroup

The minimal cardinality of an AES consisting of pure states in $\mathbb{C}^{d_1}\otimes \mathbb{C}^{d_2}$ is lower bounded by $\max(d_1,d_2)+2$ \cite{Cai2021}. It is known to be tight when $\min(d_1,d_2)=2$. In particular, any AES in $\mathbb{C}^{2}\otimes \mathbb{C}^{2}$ must consist of at least four states. The smallest AES with four elements explored in Ref.~\cite{Cai2021}, is given by
\begin{equation}\label{smallest_AES}
AES1=\{\ket{00},c\ket{00}+\sqrt{1-c^2}\ket{01},c\ket{00}+\sqrt{1-c^2}\ket{10},c\ket{00}+\sqrt{1-c^2}\ket{11}\}.
\end{equation}
This single-parameter family forms an AES for $c \in (1/2,1)$. Explicit constructions of AESs in $\mathbb{C}^{d_1}\otimes \mathbb{C}^{d_2}$ have been studied for sets containing $(d_1+d_2)$ states \cite{Cai2021}, and later extended to constructions with $d_1 d_2$ and $d_1 d_2+1$ elements \cite{Li2020arxiv}. In the following sections, we further consider additional examples of AESs related to equiangular tight frames (ETF), symmetric informationally complete positive operator-valued measures (SIC-POVM), and mutually unbiased bases (MUBs).
The notion of AES is easily extendable in the case of mixed states, and the definition is given as
\begingroup
\renewcommand{\thedefinition}{1b}
\begin{definition}
A set of $m$ mixed states $S=\{\rho_1,\cdots,\rho_m\}$, $\rho_i\in\mathcal{D}(\mathbb{C}^{d_1}\otimes\mathbb{C}^{d_2})~\forall~i$, is said to form an AES, if for every global unitary transformation $U\in U(d_1d_2)$ at least one of the transformed states $U\rho_iU^{\dagger}$ in the set remains entangled with respect to the bipartition $\mathbb{C}^{d_1}\otimes\mathbb{C}^{d_2}$. Here \(\mathcal{D}(\mathbb{C}^d)\) denotes the set of density operators acting on \(\mathbb{C}^d\).
\end{definition}
\endgroup
\setcounter{definition}{1}
For AESs consisting of mixed states, comparatively little is known. In particular, the general result that a lower bound on the minimum cardinality of a pure-state AES in $\mathbb{C}^{d_1}\otimes\mathbb{C}^{d_2}$ is $\max(d_1,d_2)+2$ has not been established for mixed-state AESs. Consequently, even in the tight two-qubit case (i.e., $d_1=d_2=2$), the existence of AESs consisting of only two or three mixed states remains an open question~\cite{Yu2021}. 
\section{Scenario and Construction of the witness in Prepare-and-Measure setup}\label{Sec_Scenario}
\begin{figure}[h!]
    \centering
    \includegraphics[width=0.5\linewidth]{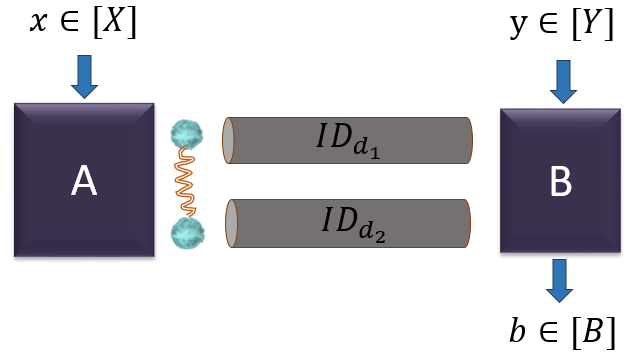}
    \caption{Two-party PM scenario in which two qudit identity channels are shared from Alice to Bob. Upon receiving an input $x\in[X]$, Alice prepares a $(d=d_1\times d_2)$-dimensional quantum state $\rho_x\in\mathcal{D}(\mathbb{C}^{d_1}\otimes\mathbb{C}^{d_2})$ and then sends the two subsystems through the respective channels to Bob. Upon receiving both subsytems, Bob performs a measurement labeled $y\in [Y]$ and outputs an outcome $b\in [B]$.}
    \label{fig1}
\end{figure}

We consider the simplest bipartite PM scenario consisting of a single sender and a single receiver, as illustrated in Fig.~\ref{fig1}. The sender, Alice, receives an input $x\in [X]$, while the receiver, Bob, receives an input $y\in [Y]$ and produces an output \(b \in [B]\). Upon receiving the input \(x\), Alice prepares a quantum state \(\rho_x\). Subsequently, depending on the input \(y\), Bob performs a measurement described by the POVM element \(M_b^y\) and returns the outcome \(b\). Therefore, the resulting PM correlations are described by
\begin{equation}
P(b|x,y)=\Tr(\rho_x M_b^y).
\end{equation}

In the PM framework, unrestricted message dimensions allow any correlation to be reproduced using a classical communication alone. Consequently, the only operational assumption imposed in our scenario is an upper bound on the message dimension, thereby placing the protocol within the semi-device-independent regime. In particular, we assume that Alice has access to two perfect qudit quantum channels, each of local dimension \(d_1\) and \(d_2\), respectively. This restriction implies that the communicated quantum systems are two-qudit states, i.e.,
\begin{equation*}
\rho_x \in \mathcal{D}(\mathbb{C}^{d_1} \otimes \mathbb{C}^{d_2}).
\end{equation*}
To certify AESs within the PM scenario, we introduce the following linear witness:
\begin{equation}
W=\sum_{b,x,y}\alpha_{bxy}P(b|x,y)\leq Q^{\mathrm{SEP}},
\end{equation}
where, the bound $Q^{\mathrm{SEP}}$ is satisfied  by all sets of separable preparations.
\begin{proposition}\label{prop1}
Consider the witness
$W=\sum_{b,x,y}\alpha_{bxy}P(b|x,y)\leq Q^{\mathrm{SEP}}$, where $Q^{\mathrm{SEP}}$ denotes the maximum value attainable using separable states. If the observed value satisfies $w>Q^{\mathrm{SEP}}$, then the set $\{\rho_x\}_x$ of the underlying states is certified to be an AES in the semi-device-independent paradigm.
\end{proposition}

\begin{proof}
Here, we proceed by contradiction by assuming the existence of a non-absolutely entangled set of states $\{\rho_x\}_x$, which violates the witness \ie $w>Q^{\mathrm{SEP}}$. By denoting the optimal measurement $y$ for this set of preparations as $\{M_b^y\}_b$ we obtain
\begin{equation}
W=\sum_{b,x,y}\alpha_{bxy}\Tr(\rho_x M_b^y)> Q^{\mathrm{SEP}}.
\end{equation}
As the set of states $\{\rho_x\}_x$ are not absolutely entangled sets there exist a global unitary $U\in U(d_1d_2)$ such that $U\rho_x U^\dagger=\sum_{\lambda}p(\lambda|x)\sigma_x^{1,\lambda}\otimes\sigma_x^{2,\lambda}~\forall~x$. By exploiting the global unitary freedom associated with the subsystem decomposition, the witness can be expressed as
\begin{align}
    &\sum_{b,x,y} \alpha_{bxy}\Tr(U\rho_x U^\dagger U{M}_b^y U^\dagger)>Q^{\mathrm{SEP}}\\
    \implies&\sum_{b,x,y} \alpha_{bxy}\Tr\sum_\lambda(p(\lambda|x)\sigma_x^{1,\lambda}\otimes\sigma_x^{2,\lambda} \widetilde{M}_b^y )>Q^{\mathrm{SEP}}, \text{where}~\widetilde{M}_b^y=UM_b^yU^\dagger.
\end{align}
This implies that even the set of separable states $\{\sum_\lambda p(\lambda|x)\sigma_x^{1,\lambda}\otimes\sigma_x^{2,\lambda}\}_x$ can also obtain the same violation. Since the witness is linear in the preparation states, its value is a convex combination of the values corresponding to the product-state components
$(\sigma_x^{1,\lambda}\otimes\sigma_x^{2,\lambda})$. Hence, if the separable ensemble achieved a value strictly larger than $Q^{\mathrm{SEP}}$, at least one corresponding product-state realization would also have to attain a value larger than $Q^{\mathrm{SEP}}$. Therefore this leads to a contradiction as the $Q^{\mathrm{SEP}}$ is the optimal violation for any separable preparations. The only assumption in this setup is the restriction on the systems local dimensions being $d_1$ and $d_2$ respectively. Hence, we proved that any violation of the witness will lead to semi-device-independent certification of absolute entangled sets.
\end{proof}

 Let $C_{\mathrm{NAES}},C_{\mathrm{SEP}},$ and $C_{\mathrm{AES}}$ denote the set of PM correlations $P(b|x,y)$ generated by non-absolutely entangled set (NAES) of states, separable states (SEP) and AES, respectively. In particular, a correlation 
 $P(b|x,y)=\Tr(\rho_xM_b^y)$ belongs to the set $C_\alpha$ with $\alpha\in\{\mathrm{NAES},\mathrm{SEP},\mathrm{AES}\}$, whenever the set of preparation ensemble $\{\rho_x\}_x$ belongs to the corresponding class $\alpha$. Let us also allow shared randomness in this scenario. Hence, if we consider two correlations $P, \tilde{P}\in C_\alpha$ then their convex combination $P_\lambda=\lambda P+(1-\lambda)\tilde{P}\in C_\alpha$. Using shared randomness, Alice and Bob implement the first preparation and measurement strategy with probability $\lambda$ and the other strategy with probability $1-\lambda$. Since both underlying strategies belong to the same class $\alpha$, and shared randomness is taken to be a free resource, the resulting correlation also belongs to $C_\alpha$. This implies that the set of correlations corresponding to each set are indeed convex. The proof of proposition \ref{prop1} exploits the global unitary freedom of the measurement choice in the prepare-and-measure setup. Therefore the PM correlations generated by any NAES can always be obtained by SEP by choosing a suitable optimal measurement. Even though the prepare-and-measure correlation $P(b|x,y)$ generated by SEP and NAES states are exactly same \ie $C_{\text{NAES}}=C_{\text{SEP}}$, the correlations generated by absolutely entangled set $C_{\text{AES}}$ is larger. The witness $W=Q^{\mathrm{SEP}}$ is a separating hyperplane between these two sets of correlations (see Fig.~\ref{fig_correlation}).

\begin{figure}[h!]
    \centering
    \includegraphics[width=0.45\linewidth]{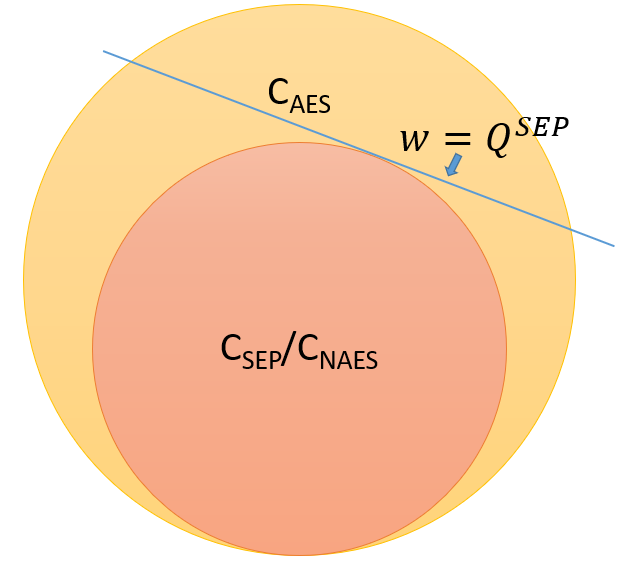}
    \caption{Illustration of the correlations generated in a two-party prepare-and-measure scenario with two qudit identity channels with local dimensions $d_1$ and $d_2$ shared between the parties. Each point represents a conditional probability vector, $\Vec{P}={P(b|x,y)}_{bxy}\in\mathbb{R}^N$. The subset $C_{\text{SEP}}=C_{\text{NAES}}$ contains correlations admitting a separable (or non-AES) realization. The hyperplane $W = Q^{\mathrm{SEP}}$ supports $C_{\text{SEP}}$. Correlations satisfying $W>Q^{\mathrm{SEP}}$ certify that every compatible set of states is an AES. The region $W\le Q^{\mathrm{SEP}}$ may nevertheless contain correlations admitting AES realizations.}
    \label{fig_correlation}
\end{figure}

\section{Methods for obtaining separable bounds}
In this section, we present the numerical methods used to bound the maximum value of a correlation witness achievable with separable preparations. As introduced above, a general linear correlation witness in the prepare-and-measure scenario can be expressed as
\begin{equation}\label{obj2}
W=\sum_{b,x,y}\alpha_{bxy}P(b|x,y)\leq Q^{\mathrm{SEP}},
\qquad
\alpha_{bxy}\in\mathbb{R},
\end{equation}
where $Q^{\mathrm{SEP}}$ denotes the maximum value attainable over all separable preparations compatible with the considered scenario. Determining $Q^{\mathrm{SEP}}$ involves a nonconvex optimization over both the preparation states and measurement operators, making its direct evaluation challenging.

To address this problem, we employ two complementary numerical approaches. First, we use an alternating (see-saw) optimization over the preparations and measurements to obtain a lower bound on $Q^{\mathrm{SEP}}$, which are often very reliable up to numerical precision. For a recent review of iterative optimization methods for quantum correlations, see Ref.~\cite{Lukacs2026}. We then introduce a hierarchy of semidefinite-programming (SDP) relaxations tailored to separable preparations, which provides systematically improved upper bounds on $Q^{\mathrm{SEP}}$. For a recent review of SDP relaxations for quantum correlations, see Ref.~\cite{Tavakoli2024}.
Together, these approaches enable us to bound the optimal separable value of a given correlation witness from both sides.
\subsection{Lower bound via alternating optimization}\label{SEESaw_section}
The witness introduced in Eq.~(\ref{obj2}) can be expressed in terms of the preparation states and measurement operators through the Born rule. The maximum value attainable with separable preparations is therefore given by
\begin{equation}
\begin{aligned}
Q^{\mathrm{SEP}}=\max_{\rho_x,M_b^y}\quad
&\sum_{b,x,y}\alpha_{bxy}\Tr(\rho_x M_b^y)\\
\text{subject to}\quad
&\rho_x \geq 0,\qquad \rho_x\in \mathrm{SEP}(\mathbb{C}^{d_1}\otimes\mathbb{C}^{d_2}),\\
&\Tr(\rho_x)=1,\qquad M_b^y \geq 0,\\
&\sum_b M_b^y = \mathbb{I},
\qquad M_b^y \in \mathcal{B}(\mathbb{C}^{d_1}\otimes\mathbb{C}^{d_2}).
\end{aligned}\label{eq_seesaw}
\end{equation}

The optimization in Eq.~\eqref{eq_seesaw} is nonconvex owing to the bilinear dependence of the objective function on the preparation states $\rho_x$ and measurement operators $M_b^y$. This difficulty can be addressed through an alternating optimization procedure, in which one set of variables is held fixed while optimizing over the other. In particular, for fixed preparation states \(\{\rho_x\}\), the optimization reduces to a semidefinite program (SDP) over the measurement operators \(\{M_b^y\}\). 
Conversely, for fixed measurements, the optimization over ${\rho_x}$ is convex apart from the separability constraint. For bipartite systems of dimensions $2\times2$ and $2\times3$, separability is equivalent to positivity under partial transposition (PPT) \cite{Peres1996,Horodecki1996}. In these cases, the preparation step can therefore be written exactly as
\begin{equation}
\begin{aligned}
\max_{{\rho_x}}\quad
&\sum_{b,x,y}\alpha_{bxy}\Tr(\rho_xM_b^y)\\
&\text{s. t.}\quad
\rho_x\geq0,\qquad \Tr(\rho_x)=1,\
&\rho_x^{T_B}\geq0,\qquad \forall x 
\end{aligned}
\end{equation}
where $T_B$ denotes partial transposition with respect to subsystem $B$. Thus, in these low-dimensional cases, both steps of the alternating optimization are SDPs and every iteration corresponds to a valid separable strategy. Starting from randomly initialized separable preparations, we alternately optimize the measurements and preparations until the objective value converges. Repeating the procedure for several random initializations and retaining the largest value obtained yields
\begin{equation}
Q_{\mathrm{LB}}^{\mathrm{SEP}}\leq Q^{\mathrm{SEP}},
\end{equation}
since the see-saw procedure is not guaranteed to reach the global optimum.

\noindent For higher-dimensional bipartite systems, the PPT criterion remains necessary but is no longer sufficient for separability in general. To obtain a feasible lower bound on the optimal separable value without relying on PPT, we employ a three-step see-saw optimization over the measurement operators and the two local factors of each product preparation (see Refs.~\cite{Vertesi2011,Bennet2014} for related three-step optimization schemes). Since the objective is linear in the preparation states, it is sufficient to optimize over pure product preparations, or more generally over product density operators of the form $\rho_x=\rho_x^A\otimes\rho_x^B.$
The corresponding optimization can be written as
\begin{equation}
Q^{\mathrm{SEP}}=
\max_{{\rho_x^A},{\rho_x^B},{M_b^y}}
\sum_{b,x,y}
\alpha_{bxy}
\Tr\left[
\left(\rho_x^A\otimes\rho_x^B\right)M_b^y
\right],
\label{eq:three_step_seesaw}
\end{equation}
subject to
\begin{align}
&\rho_x^A\geq0,\qquad \Tr(\rho_x^A)=1,\nonumber\
&\rho_x^B\geq0,\qquad \Tr(\rho_x^B)=1,\nonumber\
&M_b^y\geq0,\qquad \sum_bM_b^y=\mathbb{I}.
\end{align}
The optimization in Eq.~\eqref{eq:three_step_seesaw} is not convex, but it becomes an SDP when any two of the three sets of variables are held fixed. We therefore proceed iteratively. Starting from randomly initialized local states ${\rho_x^A}$ and ${\rho_x^B}$, we first optimize over the measurement operators ${M_b^y}$. In the second step, the measurements and Alice's local states are fixed while the optimization is performed over Bob's local states ${\rho_x^B}$. In the third step, the measurements and Bob's local states are fixed while Alice's local states ${\rho_x^A}$ are optimized. These three steps are repeated until the change in the objective value falls below a prescribed convergence threshold. To reduce the dependence on the initial conditions, the entire procedure is repeated for multiple random initializations, and the largest value obtained is retained. Since every iteration corresponds to a valid product-state preparation strategy, and hence to a valid separable strategy, the resulting value provides a lower bound on the optimal separable value, $Q_{\mathrm{LB}}^{\mathrm{SEP}}\leq Q^{\mathrm{SEP}}.$
This three-step see-saw procedure therefore provides a practical lower-bound method for higher-dimensional systems, where the PPT criterion no longer gives an exact characterization of separability.
\subsection{Upper bound by hierarchy of SDP relaxations}
The characterization of quantum correlations under Hilbert-space dimension
constraints has been extensively studied in the literature ~\cite{Navascues2015prl,Navascues2015pra,Tavakoli2019,Pauwels2022}.
To formulate the corresponding semidefinite relaxation, consider the operator set $ \mathcal{L}=\{\mathbb{I},\vec{\rho},\vec{M}\},$ where
$\vec{\rho}=(\rho_1,\rho_2,\ldots,\rho_{x})$
denotes the collection of preparation states and
$\vec{M}=(M_1^1,M_2^1,\ldots,M_{b}^{y})$
denotes the collection of measurement operators in the prepare-and-measure
scenario. Let $S_k$ denote an operator sequence consisting of all monomials
generated from $\mathcal{L}$ up to a degree $k$. The maximum degree $k$ specifies the level of the relaxation hierarchy. For each sequence $S_k$, one
constructs a positive semidefinite moment matrix $\Gamma$ of dimension
$|S_k|\times |S_k|$, with entries
\begin{equation}
    \Gamma_{\mu,\nu}=
    \Tr(\mu\nu^\dagger),
    \qquad
    \mu,\nu\in S_k.
    \label{eq:moment_matrix}
\end{equation}
Characterizing all constraints imposed on $\Gamma$ by a fixed Hilbert-space dimension is, in general, nontrivial. The hierarchy developed for dimension-bounded PM scenarios addresses this difficulty through a sampling-based construction ~\cite{Navascues2015prl,Navascues2015pra}. At a given relaxation level, preparation states $\vec{\rho}$ and measurement operators $\vec{M}$ are randomly sampled from a Hilbert space of fixed dimension $d$, and the corresponding moment matrices are constructed. By construction, these matrices satisfy the algebraic and dimension-dependent relations associated with the chosen Hilbert-space dimension. Repeated sampling generates additional moment matrices, and the procedure is terminated once newly generated matrices become linearly dependent on the previously obtained ones. This process yields a finite collection of moment matrices $\{\Gamma_1,\Gamma_2,\cdots,\Gamma_m\},$ which forms a basis for the feasible affine subspace \(\mathcal{F}\) compatible with the dimension constraints.

A candidate moment matrix can then be represented as $\Gamma=\sum_{i=1}^{m}a_i\Gamma_i,~ a_i\in\mathbb{R},$
and the desired dimension-constrained bound is obtained by imposing $\Gamma\geq0$ and optimizing the relevant linear functional $F(\Gamma)$.
The resulting semidefinite relaxation takes the form
\begin{equation}
\begin{aligned}
    \max_{\{a_i\}}\quad
    &F(\Gamma)\\
    \text{subject to}\quad
    &\Gamma=\sum_{i=1}^{m}a_i\Gamma_i,\\
    &\Gamma\geq0,\\
    & \sum_i a_i=1.
\end{aligned}
\label{eq:NV_SDP}
\end{equation}
The optimum of this relaxation provides an upper bound on the value attainable by quantum systems of the prescribed
dimension. Enlarging the operator sequence can tighten relaxations.

A relaxation of degree $k$ is defined by an operator sequence $S_k$ containing all monomials up to degree $k$. Intermediate relaxation levels can also be considered by augmenting a lower-order sequence with selected higher-order monomials. For example, at level 1+AB, the operator sequence contains all first-order monomials together with the specific second-order monomial AB. Although the random-sampling procedure used to construct the corresponding moment-matrix basis is conceptually straightforward, its computational cost increases rapidly with the relaxation level. In particular, both the number of monomials and the dimension of the associated moment matrices grow quickly as higher-order operators are included. Symmetries of the optimization problem can be exploited to substantially reduce this computational cost~\cite{Aguilar2018,Tavakoli2019}.

A direct application of the sampling based SDP construction to the present prepare-and-measure scenario yields bounds for arbitrary quantum preparations of the prescribed dimension and therefore does not distinguish separable from
entangled preparations. A natural modification is to  generate the moment-matrix basis by sampling only product preparations,
\begin{equation}
    \rho_x
    =
    \rho_{A_1}^{x}\otimes\sigma_{A_2}^{x},
    \qquad x\in[X].
    \label{eq:product_sampling}
\end{equation}
For linear witnesses, this restriction entails no loss of generality when
optimizing over separable preparations, as the optimality is always achieved by pure product states.

Nevertheless, sampling only the global product operators in
Eq.~\eqref{eq:product_sampling} does not by itself impose sufficient linear
constraints on the resulting moment-matrix space: the span generated in this
manner can coincide with that obtained from unrestricted preparations.
Consequently, the corresponding relaxation does not need to retain the desired
separability restriction and may reproduce the bound associated with
unrestricted quantum preparations.

To overcome this limitation, we introduce a modified moment-matrix
construction that explicitly retains the local operator structure underlying
each product preparation. Instead of including only the global operators
$\rho_{A_1}^{x}\otimes\sigma_{A_2}^{x}$, we augment the generating operator
set with the corresponding local operators. For each preparation input $x$, we therefore include $\rho_{A_1}^{x}\otimes\mathbb{I}_{A_2},~\mathbb{I}_{A_1}\otimes\sigma_{A_2}^{x},\rho_{A_1}^{x}\otimes\sigma_{A_2}^{x}$.
Equivalently, the preparation sector of the extended operator collection is
\begin{align}
    \vec{\rho}_{\mathrm{ext}}
    =\big(&\rho_{A_1}^{1}\otimes\mathbb{I}_{A_2},\ldots,\rho_{A_1}^{|X|}\otimes\mathbb{I}_{A_2},~\mathbb{I}_{A_1}\otimes\sigma_{A_2}^{1},\ldots,\mathbb{I}_{A_1}\otimes\sigma_{A_2}^{|X|},~\rho_{A_1}^{1}\otimes\sigma_{A_2}^{1},\ldots,\rho_{A_1}^{|X|}\otimes\sigma_{A_2}^{|X|}
    \big).
    \label{eq:extended_sep_set}
\end{align}
The inclusion of the local factors introduces additional linear relations among the entries of the moment matrix that encode the product structure of the sampled preparations. Sampling moment matrices from this enlarged operator set therefore allows the product-state structure relevant to the separable optimum to be retained within the dimension-constrained relaxation. Enlarging the operator sequence incorporates additional constraints and may produce tighter upper bounds. For the numerical implementation, we adapt the \texttt{QDimSum} package of Ref.~\cite{Tavakoli2019} to incorporate the modified sampling procedure described above. The symmetry-reduction techniques implemented in \texttt{QDimSum} substantially reduce the size of the resulting semidefinite programs. In the following section, we apply this hierarchy together with the see-saw lower-bound method discussed above to determine tight separable bounds for the correlation witnesses considered in this work.
\section{RAC witness and detection of AES}
We consider a semi-device-independent prepare-and-measure scenario in which Alice has access to two noiseless identity channels of local dimension two \ie $d_1=d_2=2$. The input received by Alice is a pair of dits \(x=(x_0,x_1)\in\{0,1,2,3\}^{2}\), while Bob receives an input \(y\in\{0,1\}\) and produces an outcome \(b\in\{0,1,2,3\}\). The most general correlation witness in this setting can be expressed as in Eq.~(\ref{obj2}). Here, we focus on a random access code (RAC) witness of the form
\begin{align}\label{RAC_witness}
W^{\mathrm{RAC}}&=\sum_{b,x,y}\alpha_{bxy}P(b|x,y)\leq Q^{\mathrm{SEP}},\\
\alpha_{bxy}&=
\begin{cases}
\frac{1}{32}, & \text{if } b=x_y,\\
0, & \text{otherwise}.
\end{cases}\nonumber
\end{align}

In other words, a positive payoff is obtained whenever Bob correctly guesses the $y$-th bit of Alice’s input string, i.e., $b=x_{y-1}$. Now the aim is to obtain a suitable bound $Q^{\mathrm{SEP}}$, such that it is satisfied by all separable preparations.

\noindent\textbf{Separable bound:}
We determine the separable bound using both alternating convex optimization (lower bound) and a semidefinite-programming (SDP) relaxation hierarchy (upper bound). As the local dimension of the systems are $d_1=d_2=2$, the PPT criterion is sufficient to impose the separability constraints. The numerical results are summarized in Table~\ref{tab:RAC_bounds}.

\begin{table}[h!]
\centering
\begin{tabular}{|c|c|cc|ccc|}
\hline
\multirow{3}{*}{Witness}
& \multirow{3}{*}{Lower bound}
& \multicolumn{2}{c|}{Upper bound: projective measurements}
& \multicolumn{3}{c|}{Upper bound: POVMs}
\\
\cline{3-7}
&
& \multicolumn{2}{c|}{Hierarchy level}
& \multicolumn{3}{c|}{Hierarchy level}
\\
\cline{3-7}
&
& 2
& $~~2+\mathrm{MMM}$
& 2
& $2+\mathrm{MMM}~$
& $~~2+\mathrm{MMM}+\rho\mathrm{MM}$
\\
\hline\hline
$W^{\mathrm{RAC}}$ & 0.728553 & 0.728905 & 0.728553 & 0.730432 & 0.730432 & 0.728553 \\\hline
\end{tabular}
\caption{The numerical lower bound is obtained using alternating optimization, while the upper bounds are derived from the SDP relaxation hierarchy. Intermediate relaxation levels are constructed by augmenting a lower-order operator sequence with selected higher-order monomials. For example, the $2+\mathrm{MMM}$ level contains all monomials up to degree 2 together with the third-order monomials generated solely from the measurement operators. For projective measurements, the upper bound matches the the lower bound at $2+\mathrm{MMM}$ level up to six decimal places. For general POVMs, agreement occurs at $2+\mathrm{MMM}+\rho\mathrm{MM}$ level, thereby identifying the optimal separable value. The SDP relaxation hierarchy is implemented using the \texttt{QDimSum} package.}
\label{tab:RAC_bounds}
\end{table}

The optimal violation of the witness $W^{\mathrm{RAC}}$ can be related to the $4$-level random access code (RAC) protocol \cite{Tavakoli2015}. Alice prepares an entangled state \(\ket{\psi}_{AB}\) and encodes her input $x_0x_1$ in the state $\ket{\psi}_{x_0x_1}$ via unitary operations as
\begin{subequations}
\begin{align}
&\ket{\psi}_{AB}=\frac{\sqrt{3}}{2}\ket{00}+\frac{1}{2\sqrt{3}}(\ket{01}+\ket{10}+\ket{11}),\\
&\ket{\psi}_{x_0x_1}=X^{x_0}Z^{x_1}\ket{\psi}_{AB},
\end{align}
\end{subequations}
where $X$ and $Z$ are the permutation and phase operations given in the computational basis as 
\begin{subequations}
\begin{align}
X=&\ket{01}\bra{00}+\ket{10}\bra{01}+\ket{11}\bra{10}+\ket{00}\bra{11},\\
Z=&\ket{00}\bra{00}+i\ket{01}\bra{01}-\ket{10}\bra{10}-i\ket{11}\bra{11}.  
\end{align}
\end{subequations}

Bob performs either a computational basis measurement or a Fourier basis measurement depending on his input $y$:
\begin{equation}
\begin{aligned}
M_1 &= \{M_{0|1}=P(\ket{00}),~M_{1|1}=P(\ket{01}),~M_{2|1}=P(\ket{10}),~M_{3|1}=P(\ket{11})\},\\
M_2 &= \Big\{M_{0|2}=P\big(\tfrac{1}{2}(\ket{00}+\ket{01}+\ket{10}+\ket{11})\big), ~M_{1|2}=P\big(\tfrac{1}{2}(\ket{00}+i\ket{01}-\ket{10}-i\ket{11})\big),\\
&\qquad M_{2|2}=P\big(\tfrac{1}{2}(\ket{00}-\ket{01}+\ket{10}-\ket{11})\big),~M_{3|2}=P\big(\tfrac{1}{2}(\ket{00}-i\ket{01}-\ket{10}+i\ket{11})\big)\Big\}
\end{aligned}
\end{equation}
where $P(\ket{\psi})=\ket{\psi}\bra{\psi}$ is the projector corresponding to $\ket{\psi}$. For this strategy, the set of preparations \(S=\{\ket{\psi}_{x_0x_1}\}_{x_0x_1}\) achieves a value \(3/4>Q^{\mathrm{SEP}}\), thereby certifying absolute entanglement in the set of the prepared states.

An experimental implementation of the corresponding RAC with ququart messages (four-level quantum systems) has already been reported, yielding an average success probability of $Q_{\mathrm{exp}}=0.7455\pm0.0020$ \cite{Miao2022}. This value exceeds the separable bound $Q^{\mathrm{SEP}}=\frac{1}{4}\left(1+\frac{1}{\sqrt{2}}\right)^2\simeq 0.728553$, and therefore violates the separable benchmark relevant to the present framework. Some care is required in interpreting this result, since the experiment employs a single four-level quantum system rather than two physically distinct qubits. Nevertheless, the Hilbert-space isomorphism $\mathbb{C}^{4}\cong\mathbb{C}^{2}\otimes\mathbb{C}^{2}$,
implies that the same prepare-and-measure statistics can, in principle, be reproduced within a two-qubit realization. Thus, although the experiment does not directly implement a bipartite two-qubit system, it demonstrates that the observed correlations are compatible with a two-qubit representation that also violates the corresponding separable bound.

On the other hand, the set of 16 product states saturating the separable bound is given by 
\begin{equation}
\ket{\chi}_{x_0x_1}=(\sigma_z\otimes P)^{x_1}(\sigma_x^{a_0}\otimes\sigma_x^{a_1})\ket{\tilde{\chi}\tilde{\chi}},
\end{equation}
where \(\sigma_x\), \(\sigma_z\), and \(P=\mathrm{diag}(1,i)\) denote the Pauli \(X\), Pauli \(Z\), and phase operators, respectively, and $\ket{\tilde{\chi}}=\cos\left(\frac{\pi}{8}\right)\ket{0}+\sin\left(\frac{\pi}{8}\right)\ket{1},$ with \(x_0=2a_0+a_1\) in binary representation (where $a_0,a_1\in\{0,1\}$). The exact separable bound obtained by these separable preparations and the same measurements is given by $Q^{\mathrm{SEP}}=\frac{1}{4}\left(1+\frac{1}{\sqrt{2}}\right)^2.$ 
Now, we consider a one-parameter family of AESs defined as
\begin{equation}
\ket{\phi}_{x_0x_1}=\frac{1}{N_{x_0x_1}(c)}\left[c\ket{\psi_{x_0x_1}}+F_{x_0+1,x_1+1}\sqrt{1-c^2}\ket{\chi_{x_0x_1}}\right],
\end{equation}
where $N_{x_0x_1}(c)$ denotes the normalization factor, and $F$ is the $4\times4$ complex phase matrix given by
\begin{equation*}
F=
\begin{pmatrix}
1 & 1 & 1 & 1\\
1 & 1 & i & 1\\
i & -i & 1 & -1\\
i & -1 & i & 1
\end{pmatrix}.
\end{equation*}
\begin{figure}[t!]
    \centering
    \includegraphics[width=0.8\linewidth]{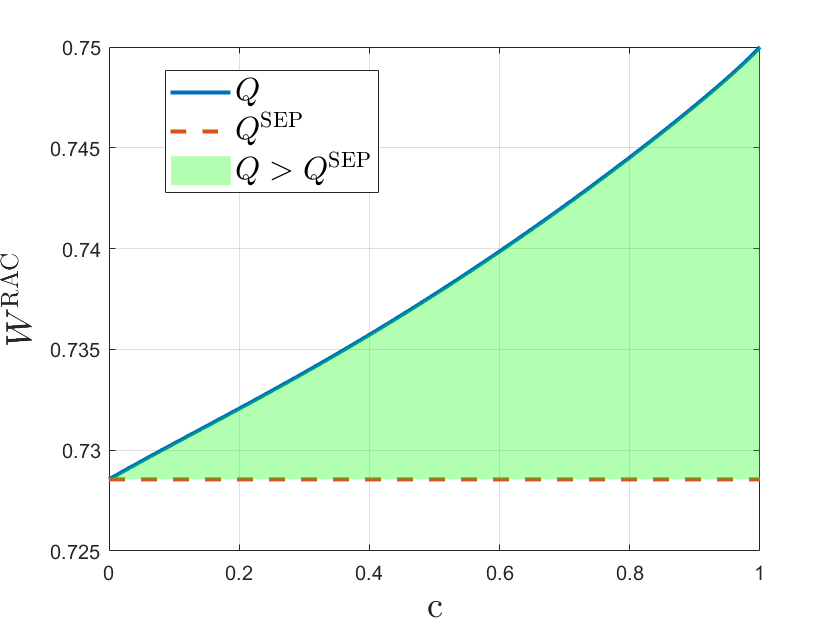}
    \caption{Random access code witness $W^{\text{RAC}}$ as a function of the state parameter $c$. The quantity $Q$ denotes the evaluated witness value for each $c$, while the separable bound $Q^{\text{SEP}}$ serves as the reference threshold. Violation of $Q^{\text{SEP}}$ certifies absolutely entangled sets in the range of $c\in(0,1]$.}
    \label{fig6}
\end{figure}
We analyze the violation of the RAC witness \(W^{\mathrm{RAC}}\) as a function of the parameter \(c\in\left[0,1\right]\). The optimal violation $Q$ of the witness [Eq.~(\ref{RAC_witness})] for the fixed preparation $\ket{\phi}_{x_0x_1}$ is obtained by using semidefinite programming. The variation of the resulting violation, $Q$ and the separable bound $Q^{\mathrm{SEP}}$ with the single parameter $c$ is shown in Fig.~\ref{fig6}. This single parameter family can be certified to be absolutely entangled in the parameter range $c\in(0,1]$.
\section{Witness Construction for Extended AESs}
The witness \(W^{\mathrm{RAC}}\) is highly effective for the certification of absolutely entangled sets (AES): any set of preparations achieving a witness value
\begin{align*}
    W^{\mathrm{RAC}}>\frac{1}{4}\left(1+\frac{1}{\sqrt{2}}\right)^2\equiv Q^{\mathrm{SEP}}
\end{align*}
can be certified as an AES. However, the converse does not generally hold. In particular, there exist AESs of pure states that do not violate the witness bound associated with \(W^{\mathrm{RAC}}\), and therefore cannot be detected using this witness alone. To overcome this limitation, in this section we develop a general methodology for constructing PM witnesses for extended preparation ensembles containing an arbitrary AES of pure states.
\subsection{General methodology}
Let us consider an AES target $S=\{\ket{\psi_1},\ket{\psi_2},\cdots,\ket{\psi_m}\}$,
where $\ket{\psi_i}\in\mathbb{C}^{d_1}\otimes\mathbb{C}^{d_2}$.
We aim to construct a suitable PM witness to certify the absolute entanglement associated with the set \(S\). It has recently been shown that arbitrary sets of preparations and measurements in finite dimension can be self-tested within the prepare-and-measure framework \cite{Navascues2023}. We employ this self-testing construction for dimension $d=d_1d_2,$ using the isomorphism $\mathbb{C}^{d_1d_2}\cong\mathbb{C}^{d_1}\otimes\mathbb{C}^{d_2}$. The target set \(S\) requires the satisfaction of the following relation:
\begin{equation}\label{condition1}
\sum_{i=1}^m\alpha_iP(\psi_i)=\frac{\mathbb{I}}{d_1d_2},
\end{equation}
for some coefficients \(\alpha_i\in\mathbb{R}^+\), where \(P(\psi_i)=\ket{\psi_i}\bra{\psi_i}\) denotes the projector onto \(\ket{\psi_i}\). If the target states does not satisfy Eq. \eqref{condition1}, it can always be augmented with a suitable collection of ancillary states to fulfill the condition. In addition, the self-testing construction requires the enlarged ensemble to possess the Wigner property. When this property is not already satisfied, one may further append a suitable set of fiducial states, following the construction of Lemma~4 of Ref.~\cite{Navascues2023}, so that the resulting ensemble fulfills the required conditions. Consequently, we consider the extended preparation set $\widetilde{S}=\{\ket{\psi_1},\ket{\psi_2},\cdots,\ket{\psi_m},\ket{\psi_{m+1}},\cdots,\ket{\psi_N}\}$, where $\{\ket{\psi_{m+1}},\ket{\psi_{m+2}},\cdots,\ket{\psi_N}\}$ denotes collectively the ancillary states and when required, the fiducial states also. An important limitation of this construction is that the resulting witness does not, in general, certify that the original target set $S$ itself is an AES. Rather, the certification applies to the extended set $\widetilde{S}$ obtained by augmenting $S$ with the ancillary states. Thus, a violation of the corresponding separable bound certifies the existence of absolute entanglement at the level of the extended preparation set $\widetilde{S}$, but does not by itself imply that $S$ is an AES.

The resulting witness is implemented within the prepare-and-measure scenario introduced in Sec.~\ref{Sec_Scenario}, with \(|X|=N\). Bob’s input is a pair, $y=f_{ij}=(i,j),~ i\neq j,$ leading to $|Y|=\binom{N}{2},$ and for each input Bob performs a two outcome measurement $M_{f_{ij}}=\{\Pi^1_{f_{ij}},\Pi^2_{f_{ij}}\}$, \ie , $|B|=2.$
To certify that the extended set $\widetilde{S}$ is an AES, we consider the modified witness
\begin{equation}
\begin{aligned}
W(\widetilde{S})
=
\sum_{1\le j<i\le N}
\alpha_i\alpha_j
{\|P(\psi_i)-P(\psi_j)\|}_1
\Big(
&P(b=2|x=i,y=f_{ij})
\\
-&P(b=2|x=j,y=f_{ij})
\Big)
\leq Q^{\mathrm{SEP}}.
\end{aligned}
\end{equation}

The bound \(Q^{\mathrm{SEP}}\) is chosen such that it is satisfied by all separable preparations. In general, deriving this bound analytically is challenging; therefore, we employ numerical optimization techniques to estimate it reliably. Specifically, we first use alternating convex optimization to obtain a lower bound \(Q^{\mathrm{SEP}}_{\mathrm{LB}}\), followed by a semidefinite-programming relaxation hierarchy to derive an upper bound \(Q^{\mathrm{SEP}}_{\mathrm{UB}}\). The target preparation set \(\widetilde{S}\) achieves the optimal violation value \(1-\frac{1}{d_1d_2}=3/4\), as the witness self-tests the states in \(\widetilde{S}\). Consequently, if an observed violation \(Q\) satisfies $Q^{\mathrm{SEP}}_{\mathrm{LB}}<Q<Q^{\mathrm{SEP}}_{\mathrm{UB}},$ this provides strong numerical evidence for the presence of absolute entanglement in the preparation set. A rigorous certification of AES is obtained whenever $Q>Q^{\mathrm{SEP}}_{\mathrm{UB}}.$

In principle, tighter upper bounds can be obtained by increasing the level of the relaxation hierarchy. However, exploring higher levels is often computationally demanding due to the limited symmetry structure of the witness and the rapid growth in numerical complexity. In the following sections, we demonstrate the applicability of this methodology through several explicit examples of AES constructions.

\subsection{AES I: ETF-Based Construction}
We aim to utilize the symmetry of the self-testing witness to reduce the computational complexity. This leads us to the equiangular tight frame (ETF) being the best suitable candidate for AES certification. A set of $m$ states $S'=\{\ket{\psi_1},\ket{\psi_2},\cdots,\ket{\psi_m}\}$ in the Hilbert space $\mathbb{C}^d$ forms an ETF if they satisfy the following conditions (i) $|\braket{\psi_i|\psi_j}|=c~\forall ~i\neq j$, (ii) $\sum_{i=1}^m\ket{\psi_i}\bra{\psi_i}=\frac{m}{d}\mathbb{I}$ \cite{Welch1974,Sustik2007}. Hence, the self-testing witness used for entanglement certification takes the following symmetric form
\begin{align}\label{ETF1}
    W(S')=&~\tilde{c}\sum_{1\leq i< j\leq m}P(b=2|x=i,y=f_{ij})-P(b=2|x=j,y=f_{ij})\leq Q^{\mathrm{SEP}}\nonumber\\
    =&~\tilde{c}\sum_{1\leq i< j\leq m}\Tr [(\ket{\psi_i}\bra{\psi_i}-\ket{\psi_j}\bra{\psi_j})\Pi_{f_{ij}}^2]\leq Q^{\mathrm{SEP}}
\end{align}
where $\tilde{c}=(2/m^2)\sqrt{1-c^2}$ and $\Pi^2_{f_{ij}}$ is the second POVM element of the measurement $M_{f_{ij}}=\{\Pi^1_{f_{ij}},\Pi^2_{f_{ij}}\}$.

To express the symmetry of the witness, we introduce a binary outcome observable $A_{f_{ij}}$ corresponding to the measurement $M_{f_{ij}}$ as $A_{f_{ij}}=\Pi^2_{f_{ij}}-\Pi^1_{f_{ij}}=2\Pi^2_{f_{ij}}-\mathbb{I}$. The self-testing witness can be expressed in terms of the observable $A_{f_{ij}}$ as
\begin{align*}
    W(S')=&~\frac{\tilde{c}}{2}\sum_{1\leq i< j\leq m}\Tr [(\ket{\psi_i}\bra{\psi_i}-\ket{\psi_j}\bra{\psi_j})A_{f_{ij}}]\leq Q^{\mathrm{SEP}}
\end{align*}
The witness remains unchanged under the following transformation 
\begin{subequations}
\begin{align}
    &\rho_l\longleftrightarrow\rho_m\\
    &A_{f_{lm}}\rightarrow -A_{f_{lm}}~~~~~~~\forall~ l<m, \\
    &A_{f_{li}}\leftrightarrow -A_{f_{im}}~~~~~~~~if ~~l<i<m\\
    &A_{f_{li}}\leftrightarrow A_{f_{mi}}~~~~~~~~~~if ~i>m
\end{align}\label{sym1}
\end{subequations}
We exploit these symmetries to obtain the upper bound $Q^{\mathrm{SEP}}$ corresponding to the separable preparations. The smallest overcomplete ETF in $\mathbb{C}^2\otimes\mathbb{C}^2$ is a regular simplex of five states, but unfortunately this set does not define an AES. As there does not exist any six-state ETF \cite{Wei2024}, we have considered the ETF set consisting of the following seven states
\begin{align}
    S_1&=\{\ket{\psi_j}\}_{j=1}^7\nonumber\\
    \ket{\psi_{j}}&=\frac{1}{2}(\ket{00}+\omega^{j-1}\ket{01}+\omega^{2j-2}\ket{10}+\omega^{4j-4}\ket{11})\\
    &\text{where } ~\omega=e^{2\pi i/7}~,~c=\frac{1}{2\sqrt{2}}\nonumber.
\end{align}
The corresponding witness for the certification of the absolute entanglement in the set is given as 
\begin{align}
    W(S_1)=\frac{1}{7\sqrt{14}}\sum_{1\leq i< j\leq 7}P(b=2|x=i,y=f_{ij})-P(b=2|x=j,y=f_{ij})\leq Q^{\text{SEP}}
\end{align}
The numerical estimation of the separable bound $Q^{\text{SEP}}$ is summarized in Table~\ref{tabETF1}.
\begin{table}[h!]
    \centering
    \begin{tabular}{|c|c|c|c|c|c|c|}
    \hline
        \multirow{2}{*}{Witness} &
        \multirow{2}{*}{Lower bound} &
        \multicolumn{5}{c|}{Upper bound at level}\\
        \cline{3-7}
        & & 2 & 2+$\rho\rho\rho$ & 2+$\rho\rho\rho+\rho\rho\mathrm{M}$ & 2+$\rho\rho\rho+\rho\rho \mathrm{M}+\mathrm{M}\rho\rho$ & 2+$\rho\rho\rho+\rho\rho \mathrm{M}+\rho \mathrm{MM}$\\
        \hline
        $W(S_1)$ & 0.749801 & 3/4 & 0.749952 & 0.749935 & 0.749935 & 0.749935\\
        \hline
    \end{tabular}
    \caption{Numerical lower and upper bounds for the separable witness value associated with \(S_1\). The SDP relaxation hierarchy is implemented using the \texttt{QDimSum} package.}
    \label{tabETF1}
\end{table}
The target states, $S_1$ achieve the optimal value $Q=3/4$, thereby certifying the existence of absolute entanglement in $S_1$. 

\subsection{AES II: SIC POVM based construction - Maximal ETF}

Symmetric informationally complete positive operator-valued measures (SIC POVMs) constitute extremal measurements in which all elements exhibit equal pairwise overlaps. The ETF with maximal elements in $\mathbb{C}^d$ corresponds to the  \(d^2\) subnormalized rank-one operators of SIC POVM. The existence of SIC POVM in higher dimensions has been studied both analytically \cite{Appleby2005,Grassl2005} and numerically \cite{Renes2004,Scott2010}.

Here, we consider the set of 16 states associated with a SIC POVM in the Hilbert space \(\mathbb{C}^2\otimes\mathbb{C}^2\) denoted as $S_2$ \cite{Zhu2010}. Since these states possess equal overlaps, i.e., $|\braket{\psi_i|\psi_j}|=\frac{1}{\sqrt{5}},$ it follows that the corresponding difference satisfies $\|P(\psi_i)-P(\psi_j)\|_1=\frac{4}{\sqrt{5}}$. Moreover, the self-testing normalization condition can be satisfied by choosing $\alpha_i=\frac{1}{16},~\forall i.$  Accordingly, the self-testing witness for the SIC-POVM-based set takes the form
\begin{equation}
W(S_2)
=
\sum_{1\le j<i\le 16}
\left(\frac{1}{16}\right)^2
\frac{4}{\sqrt{5}}
\bigl(
P(b=2|x=i,y=f_{ij})
-
P(b=2|x=j,y=f_{ij})
\bigr)
\leq Q^{\mathrm{SEP}}.
\end{equation}

The numerical estimates of the separable bound \(Q^{\mathrm{SEP}}\) are summarized in Table~\ref{tabAES3}.

\begin{table}[h!]
    \centering
    \begin{tabular}{|c|c|c|c|c|}
    \hline
        \multirow{2}{*}{Witness} &
        \multirow{2}{*}{Lower bound} &
        \multicolumn{3}{c|}{Upper bound at level}\\
        \cline{3-5}
        & & $\rho+\mathrm{M}+\rho \mathrm{M}$ &$\rho+\mathrm{M}+\rho \mathrm{M}+ \rho\rho$& $2+\rho\rho\rho$\\
        \hline
        $W(S_2)$ &  0.748318 & 0.788870 & 0.75 & 0.75\\
        \hline
    \end{tabular}
    \caption{Numerical lower and upper bounds for the separable witness value associated with the SIC-POVM-based AES \(S_2\).}
    \label{tabAES3}
\end{table}
We have utilized the ETF symmetries Eq.~(\ref{sym1}), but still the upper bound does not converge below the quantum optimal value $3/4$ for level $2+\rho\rho\rho$ with \texttt{QDimSum} \cite{Tavakoli2019}. It becomes computationally hard to check for higher levels of the relaxation hierarchy. Still, the \(16\) SIC-POVM states achieve the optimal value \(Q=3/4>0.7483189\equiv Q^{\mathrm{SEP}}_{\text{LB}}\). This small gap provides strong numerical evidence, but not a rigorous certification, of their absolute entanglement in the prepare-and-measure scenario.
\subsection{AES III: Construction from a minimal AES}
We next consider the smallest known AES in
$\mathbb{C}^{2}\otimes\mathbb{C}^{2}$, which consists of four pure
states~\cite{Cai2021}. The corresponding construction forms a
one-parameter family introduced in Eq.~(\ref{smallest_AES}). Here, we focus on the particular instance $c=1/\sqrt{2}$,
\begin{equation}
S_3=
\left\{
\ket{\psi_1}=\ket{0}\otimes\ket{0},
\ket{\psi_2}=\ket{0}\otimes\ket{+},
\ket{\psi_3}=\ket{+}\otimes\ket{0},
\ket{\psi_4}=\frac{1}{\sqrt{2}}(\ket{00}+\ket{11})
\right\}.
\label{eq:smallest_AES_instance}
\end{equation}

To construct the corresponding self-testing witness, we augment the target
set with ancillary preparations such that the resolution-of-the-identity
condition in Eq.~(\ref{condition1}) is satisfied. We choose
\begin{align}
\alpha_i
&=
\frac{1}{6+2\sqrt{7}},
\qquad
i=1,2,3,4,
\nonumber\\
\alpha_5=\alpha_6
&=
\frac{\sqrt{7}-1}{8},
\qquad
\alpha_7
=
\frac{3\sqrt{7}-7}{4},
\label{eq:smallest_AES_coefficients}
\end{align}
together with the ancillary states
\begin{align}
\ket{\psi_5}
&=
\frac{1}{\sqrt{6}}\ket{01}
-\sqrt{\frac{2}{3}}\ket{10}
+\frac{1}{\sqrt{6}}\ket{11},
\nonumber\\
\ket{\psi_6}
&=
\ket{-}\otimes\ket{1},
\nonumber\\
\ket{\psi_7}
&=
a_1\ket{00}
-b_1
\left(
\ket{01}+\ket{10}+\ket{11}
\right),
\label{eq:smallest_AES_ancillary}
\end{align}
where
\begin{equation}
a_1=
\sqrt{\frac{7-2\sqrt{7}}{14}},
\qquad
b_1=
\sqrt{\frac{7+2\sqrt{7}}{42}}.
\label{eq:smallest_AES_ab}
\end{equation}
The resulting extended preparation ensemble is $\widetilde{S}_3=\{\ket{\psi_1},\ldots,\ket{\psi_7}\}.$ Following the general construction introduced above, we associate with
$\widetilde{S}_3$ the witness
\begin{equation}
\begin{aligned}
W(\widetilde{S}_3)
=\sum_{1\leq j<i\leq 7}\alpha_i\alpha_j\left\|P(\psi_i)-P(\psi_j)\right\|_1\Big[P(b=2|x=i,y=f_{ij})-P(b=2|x=j,y=f_{ij})\Big].
\end{aligned}
\label{eq:smallest_AES_witness}
\end{equation}
For four-dimensional preparations, the target realization attains the
maximum quantum value $W(\widetilde{S}_3)=\frac{3}{4}.$
We first estimate the separable benchmark using the alternating see-saw optimization described in the preceding section. This yields the lower bound $Q_{\mathrm{LB}}^{\mathrm{SEP}}(\widetilde{S}_3)=0.74990679.$
The target value therefore exceeds the best separable value found by the
see-saw procedure by $\frac{3}{4}-Q_{\mathrm{LB}}^{\mathrm{SEP}}(\widetilde{S}_3)\simeq9.32\times10^{-5}.$

For the previous examples, the upper bounds were obtained by exploiting the symmetry structure of the corresponding witnesses within the \texttt{QDimSum} implementation of the SDP hierarchy. In contrast, the present witness does not exhibit the same symmetry structure, and the resulting SDP becomes computationally prohibitive within our current implementation. We are therefore unable to obtain a sufficiently tight upper bound on the separable value for this example.

Consequently, the present numerical analysis does not provide a rigorous semi-device-independent certification of the target AES $\widetilde{S_3}$. Nevertheless, the see-saw optimum lies extremely close to the maximum quantum value $3/4$, leaving only a small unresolved interval $0.74990679\leq Q^{\mathrm{SEP}}(\widetilde{S}_3)\leq\frac{3}{4}.$ Closing this gap requires a tighter SDP relaxation to obtain an upper bound on the separable value.

\subsection{AES IV: MUB-based construction}
We next consider a construction based on mutually unbiased bases (MUBs). Two orthonormal bases
\begin{equation*}
\Phi=\{\ket{\phi_1},\ket{\phi_2},\ldots,\ket{\phi_d}\},
\qquad
\Psi=\{\ket{\psi_1},\ket{\psi_2},\ldots,\ket{\psi_d}\},
\end{equation*}
of a $d$-dimensional Hilbert space are mutually unbiased if
\begin{equation}
\left|\braket{\phi_i|\psi_j}\right|^2=\frac{1}{d},\qquad\forall\,i,j.
\label{eq:MUB_definition}
\end{equation}

MUBs provide a natural construction of AESs because only a limited number
of mutually unbiased bases can simultaneously consist entirely of product
states with respect to a fixed tensor-product structure. The maximum number of product MUBs in the multipartite Hilbert space $\mathbb{C}^{d_1}\otimes\mathbb{C}^{d_2}\otimes\cdots\otimes\mathbb{C}^{d_n}$ is given by $p=\min_j M_j,$ where \(M_j\) denotes the number of MUBs in \(\mathbb{C}^{d_j}\) \cite{McNulty2016}. In particular, for the bipartite Hilbert space $\mathbb{C}^{2}\otimes\mathbb{C}^{2}$, at most three mutually unbiased bases can simultaneously be product bases. Since mutual unbiasedness is preserved under a common global unitary transformation, any collection of four MUBs in this space cannot be mapped simultaneously to product bases. The union of the states belonging to such four bases therefore constitutes an AES.

We consider the following set of four MUBs in
$\mathbb{C}^{2}\otimes\mathbb{C}^{2}$:
\begin{align}
S_4=\bigg\{&
\ket{00},\ket{01},\ket{10},\ket{11},
\ket{++},\ket{+-},\ket{-+},\ket{--},
\ket{y_0y_0},\ket{y_0y_1},
\ket{y_1y_0},\ket{y_1y_1},
\nonumber\\
&
\frac{1}{\sqrt{2}}
\left(\ket{y_0+}+i\ket{y_1-}\right),
\frac{1}{\sqrt{2}}
\left(\ket{y_0+}-i\ket{y_1-}\right),
\frac{1}{\sqrt{2}}
\left(\ket{y_1+}+i\ket{y_0-}\right),
\frac{1}{\sqrt{2}}
\left(\ket{y_1+}-i\ket{y_0-}\right)
\bigg\},
\label{eq:MUB_AES}
\end{align}
where
\begin{equation}
\ket{\pm}
=
\frac{\ket{0}\pm\ket{1}}{\sqrt{2}},
\qquad
\ket{y_{0/1}}
=
\frac{\ket{0}\pm i\ket{1}}{\sqrt{2}}.\nonumber
\end{equation}
The first three bases consist entirely of product states, whereas the
fourth contains entangled states. Since $S_4$ is the union of four orthonormal bases, no additional preparations are required to satisfy the resolution-of-the-identity condition introduced in Eq.~(\ref{condition1}). Indeed, choosing $\alpha_i=\frac{1}{16},
\quad \forall i=1,\ldots,16,$ gives
\begin{align}
\sum_{i=1}^{16}
\alpha_i\ket{\psi_i}\bra{\psi_i}
&=
\frac{1}{16}
\sum_{k=1}^{4}
\sum_{i=1}^{4}
\ket{\psi_i^{(k)}}\bra{\psi_i^{(k)}}=
\frac{1}{16}
\sum_{k=1}^{4}\mathbb{I}_4
=
\frac{\mathbb{I}_4}{4},
\label{eq:MUB_completeness}
\end{align}
as required for a four-dimensional preparation space. Following the general construction introduced above, we associate with
$S_4$ the witness
\begin{equation}
\begin{aligned}
W(S_4)=\sum_{1\leq j<i\leq16}\left(\frac{1}{16}\right)^2\left\|P(\psi_i)-P(\psi_j)\right\|_1\Big[&P(b=2|x=i,y=f_{ij})\nonumber\\
-&P(b=2|x=j,y=f_{ij})\Big].
\end{aligned}
\label{eq:MUB_AES_witness}
\end{equation}
For four-dimensional quantum preparations, the target realization attains the maximum value $W(S_4)=\frac{3}{4}.$ We estimate the separable value using the alternating see-saw optimization described in Sec.~\ref{SEESaw_section}. Since the preparations are restricted to $\mathbb{C}^{2}\otimes\mathbb{C}^{2}$, the PPT criterion provides an exact characterization of separability in each preparation step. The optimization yields $Q_{\mathrm{LB}}^{\mathrm{SEP}}(S_4)=0.748680.$ Thus, the target value exceeds the largest separable value found by the see-saw procedure by $\frac{3}{4}-Q_{\mathrm{LB}}^{\mathrm{SEP}}(S_4)=1.32\times10^{-3}.$ For all the preceding examples, we obtained upper bounds by exploiting symmetries of the corresponding preparation ensembles within the \texttt{QDimSum} implementation of the SDP relaxation hierarchy. The present MUB-based witness, however, does not exhibit sufficient exploitable symmetry within our current implementation, making the corresponding SDP computationally demanding. We therefore do not obtain a sufficiently tight upper bound on the separable value for this example.

Consequently, the present calculation and the violation of the see-saw separable bound provide heuristic evidence for the existence of absolute entanglement in the set $S_4$ within the PM scenario. However, this does not by itself provide a rigorous semi-device-independent certification of the AES $S_4$. Instead, it constrains the separable optimum to the interval $0.748680 \leq Q^{\mathrm{SEP}}(S_4) \leq\frac{3}{4},$ where the upper endpoint follows from the maximum quantum value of the witness. Closing this gap requires a tighter numerical relaxation to obtain the upper bound on the separable value.

\section{Discussion}
We have developed a semi-device-independent framework for certifying absolutely entangled sets (AESs) in a prepare-and-measure scenario. The primary objective of this work is to provide an operational approach to the certification of AESs, complementary to their conventional state-dependent characterization. To this end, we construct prepare-and-measure correlation witnesses and characterize their optimal values over separable preparations. We employ alternating see-saw optimization to obtain lower bounds on the separable value and introduce a new semidefinite-programming (SDP) relaxation hierarchy to derive corresponding upper bounds. In particular, violation of the separable bound associated with our RAC-based witness certifies that the preparation ensemble constitutes an AES. The converse, however, does not hold in general, since some AESs may fail to violate this particular witness. To address this limitation, we further develop a general methodology for constructing certification witnesses tailored to extended state ensembles containing arbitrary target pure-state AESs, based on the prepare-and-measure self-testing construction of Ref.~\cite{Navascues2023}. Entanglement is a fundamental resource in quantum information theory, and correlations generated from entangled states enable advantages in a wide range of information-processing tasks. In the present prepare-and-measure setting, however, entanglement of the individual preparations is not by itself sufficient to distinguish their correlations from those obtainable with separable preparations. In particular, preparation ensembles that are not absolutely entangled can be transformed, by an appropriate global unitary, into separable ensembles, and the corresponding unitary freedom can be absorbed into the measurement device. As a consequence, their observable correlations cannot in general be distinguished from those generated by separable preparations. By contrast, for suitable witnesses, AESs can generate correlations beyond the region accessible with separable preparations. This provides the operational basis for the certification framework developed here.

Previous studies of AESs have primarily focused on their structural and geometric properties, including the construction of explicit AESs and the characterization of conditions under which a collection of states remains entangled under arbitrary global unitary transformations ~\cite{Cai2021,Yu2021,Li2020arxiv}. Our approach instead addresses the operational certification of this property from observed prepare-and-measure correlations. The preparation and measurement devices are otherwise treated as uncharacterized, with the only assumption being a bound on the Hilbert-space dimension of the communicated system. The resulting certification is therefore semi-device-independent and does not require state tomography or a detailed characterization of the measurement apparatus. A central technical ingredient of our framework is the determination of separable bounds for dimension-constrained prepare-and-measure witnesses. Such optimization problems arise more broadly in quantum communication whenever one seeks to distinguish correlations generated by separable preparations from unrestricted quantum correlations. Related SDP relaxation techniques have previously been employed to bound correlations generated by separable states in network communication scenarios~\cite{RochiCarceller2026}. In the present work, we introduce a modified dimension-constrained moment-matrix construction tailored to separable preparations. By explicitly retaining the local operator structure of product preparations, the resulting SDP relaxations provide upper bounds on the correlations compatible with separable preparation strategies. Together with feasible values obtained from see-saw optimization, this allows the optimal separable value to be constrained from both sides.

Although the witness-construction framework applies to arbitrary bipartite systems $\mathbb{C}^{d_1}\otimes\mathbb{C}^{d_2}$, our numerical investigations have focused on the two-qubit case $\mathbb{C}^{2}\otimes\mathbb{C}^{2}$. Extending the numerical analysis to higher-dimensional systems therefore constitutes a natural direction for future work. Moreover, for the constructions based on the smallest AESs and MUBs considered here, the limited exploitable symmetry of the corresponding witnesses makes the higher levels of the SDP relaxation computationally demanding, preventing us from obtaining sufficiently tight separable upper bounds with the present implementation. Developing stronger or more efficient SDP relaxations, improved symmetry-reduction techniques, or analytical methods for bounding the separable value would therefore be valuable. More broadly, our framework suggests that absolute entanglement can be viewed not only as a geometric property of a collection of states under global unitary transformations, but also as an operational resource that can be revealed through prepare-and-measure correlations. This perspective motivates the investigation of AESs in other prepare-and-measure tasks, network communication scenarios, and higher-dimensional or multipartite settings.

\section{Acknowledgements}
We thank Jef Pauwels and Nicolas Brunner for useful discussions. We acknowledge support from the European Union (CHIST-ERA MoDIC) and from the National Research, Development and Innovation Office NKFIH (Grant Nos.~2023-1.2.1-ERA\_NET-2023-00009 and K145927).
%

\end{document}